\documentclass[sigconf]{aamas} 

\usepackage{balance} 
\usepackage{tikz}
\usetikzlibrary{positioning}
\usepackage{multirow}
\usepackage{hyperref}
\usepackage{amsthm}
\usepackage{balance} 
\usepackage[ruled,linesnumbered]{algorithm2e}
\usepackage{hyperref}       
\usepackage{url}            
\usepackage{booktabs}       
\usepackage{amsmath}       

\usepackage{amsfonts}       
\usepackage{microtype}
\newcommand{\indicator}[1]{\textbf{1}_{\left[{#1}\right]}}
\newcommand{\vValue}{\psi}
\newcommand{\shadingFunc}{\beta}

\newcommand{\Expb}[1]{\mathbb{E}\left(#1\right)}

\newcommand{\eps}{\epsilon}

\newcommand{\tendsto}{\rightarrow}
\newcommand{\Unif}{\mathcal{U}}

\usepackage{amsmath}      
\usepackage{graphicx}  
\newtheorem{theorem}{Theorem}
\newtheorem{lemma}{Lemma}

\newtheorem{definition}{Definition}

 \newcommand{\argmax}[1]{\underset{#1}{\operatorname{arg}\,\operatorname{max}}\;}

\setcopyright{ifaamas}
\acmConference[AAMAS '21]{Proc.\@ of the 20th International Conference on Autonomous Agents and Multiagent Systems (AAMAS 2021)}{May 3--7, 2021}{Online}{U.~Endriss, A.~Now\'{e}, F.~Dignum, A.~Lomuscio (eds.)}
\copyrightyear{2021}
\acmYear{2021}
\acmDOI{}
\acmPrice{}
\acmISBN{}

\acmSubmissionID{???}

\title[Adversarial Learning in Revenue-Maximizing Auctions]{Adversarial Learning in Revenue-Maximizing Auctions}

\author{Thomas Nedelec}
\affiliation{
 \institution{Criteo AI Lab, ENS Paris Saclay}}
\email{thomas.nedelec@polytechnique.org}

\author{Jules Baudet}
\affiliation{
 \institution{Ecole Polytechnique}}
\email{jules.baudet99@gmail.com}

\author{Vianney Perchet}
\affiliation{
  \institution{ENSAE, Criteo AI Lab}}
\email{vianney@ensae.fr}

\author{Noureddine El Karoui}
\affiliation{
  \institution{Criteo AI Lab, UC, Berkeley}}
\email{nkaroui@berkeley.edu}

\begin{abstract}
We introduce a new numerical framework to learn optimal bidding strategies in repeated auctions when the seller uses past bids to optimize her mechanism. Crucially, we do not assume that the bidders know which optimization mechanism is used by the seller. We recover essentially all state-of-the-art analytical results for the single-item framework  derived previously in the setup where the bidder knows the optimization mechanism used by the seller and extend our approach to multi-item settings, in which no optimal shading strategies were previously known. Our approach yields substantial increases in bidder utility in all settings and has a strong potential for practical usage since it provides a simple way to optimize bidding strategies on modern marketplaces where buyers face unknown data-driven mechanisms.  

\end{abstract}

\keywords{Auction Theory; Adversarial Learning; Strategic bidder}

\newcommand{\BibTeX}{\rm B\kern-.05em{\sc i\kern-.025em b}\kern-.08em\TeX}

\begin{document}


\pagestyle{fancy}
\fancyhead{}


\maketitle 

\section{Introduction} 

Repeated auctions are widely used in modern economic systems to sell a variety of items such as ad placements on the Internet. In online marketplaces, most auctions are designed using techniques at the junction of classical auction theory \cite{Myerson81} and statistical learning theory. Sellers take advantage of the enormous amount of data  gathered on buyers' behavior and strategies - through billions of auctions a day - to learn and implement revenue maximizing auctions on different platforms.

In the case of single-item auctions, the design of an optimal incentive-compatible revenue-maximizing auction is well understood \cite{Myerson81}, assuming  the seller knows the value distribution of each buyer.  Indeed, under this perfect knowledge assumption, she can define the allocation and payment rules maximizing her expected revenue. The multi-item framework is more intricate. Myerson's fundamental result has been extended to specific settings depending on the number of objects and on the properties of the bidders' utility functions  \cite{armstrong1996multiproduct,manelli2007multidimensional,daskalakis2013mechanism,yao2017dominant}.  A general and analytical optimal auction in the multi-item framework has yet to be found. 

Because of the large variety of different multi-item settings, \emph{automatic mechanism design} has been introduced to provide a framework for learning revenue-maximizing mechanisms satisfying constraints chosen by the designer \cite{conitzer2002complexity,albert2017automated}. This framework was recently complemented by the introduction of neural networks  for different instances of the multi-item problem \cite{dutting2017optimal,shen2019automated,golowich2018deep} to take advantage of the large expressivity power of neural networks architectures.

This line of research traditionally  assumes that bidders' value distributions are known to the seller. However, in practice, the seller does not have access to such information and can only statistically estimate these distributions using a finite sample of bids made in past auctions \cite{OstSch11,cole2014sample,morgenstern2015pseudo,huang2018making,balcan2018general}. We can represent this as a two-stage game between a seller and buyers. The first stage consists in a sequence of truthful auctions (say, second price auctions without reserve price or with random reserve prices) where  bidders are assumed to bid truthfully. This will provide the seller with a batch of i.i.d. samples  from the different value distributions (since, in truthful auctions, observed bids are equal to unobserved values). Under this truthfulness assumption, the seller can compute the empirical  revenue-maximizing auction, based on the bid samples collected in the first stage \cite{paes2016field,medina2017revenue,shen2019learning,derakhshan2019lp}. 

\begin{figure*}[t]
\center
\begin{tikzpicture}[
    block/.style={rectangle, draw, text centered, rounded corners,     minimum
    height=4em,text width=5.5em}
]
	
	\node[thick] (z) {Strategic bidder};
	\node[above=of z, thick] (xinit) {Dynamic game};
	[
	\node[above=of xinit, thick] (xinit2) {Strategic seller};

	\node[rectangle, draw, thick, right=1.0em of xinit, text width=1.8cm] (bidderbox) {$m$ auctions of type $M_0$} ;
	
	\node[circle, draw, thick, right=1.5em of bidderbox, text width=0.6cm] (dataset) {$S_m^0$} ;
	
	\node[rectangle, draw, thick, right=2.0em of dataset, text width=1.8cm] (mech1) {$m$ auctions of type $\mathcal{M}_1(S_m^0)$} ;

	\node[circle, draw, thick, right=2.0em of mech1, text width=0.6cm] (dataset2)  {$S_m^1$} ;

	\node[rectangle, draw, thick, right=2.0em of dataset2, text width=1.8cm] (mech2) {$m$ auctions of type $\mathcal{M}_2(S_m^1)$} ;

	\node[ thick, right=1.0em of mech2, text width=1.5cm] (end) {$\dots$} ;
	
	\draw[-stealth] (bidderbox) --  (dataset);
	\draw[-stealth] (dataset) --  (mech1);
	\draw[-stealth] (mech1) --  (dataset2);
	\draw[-stealth] (dataset2) --  (mech2);
	\draw[-stealth] (mech2) --  (end);

	\node[right=0.5em of dataset, circle, fill, inner sep=0.15em] (pt2) {};
	\node[right=0.5em of dataset2, circle, fill, inner sep=0.15em] (pt4) {};
	
	\node[rectangle, draw, thick, below= of bidderbox] (bidderstrat1) {$\beta_{i,0}$} ;
	\node[rectangle, draw, thick, above= of pt2] (seller1) {$\mathcal{M}_1$} ;
	\node[rectangle, draw, thick, above= of pt4] (seller2) {$\mathcal{M}_2$} ;
	\node[rectangle, draw, thick, below= of mech1] (bidderstrat2) {$\beta_{i,1}$} ;
	\node[rectangle, draw, thick, below= of mech2] (bidderstrat3) {$\beta_{i,2}$} ;

	\draw[-stealth] (bidderstrat1) --  (bidderbox);
	\draw[-stealth] (bidderstrat2) --  (mech1);
	\draw[-stealth] (bidderstrat3) --  (mech2);

	\draw[-stealth] (seller1) --  (pt2);
	\draw[-stealth] (seller2) --  (pt4);
			
			
			

\end{tikzpicture}
\caption{\textbf{A general dynamic game for the batch auction setting.} During the first stage of the game, the seller, as she has no information on bidders' value distribution, runs a batch of $m$ second-price auctions without reserve price (mechanism $M_0$). The strategic bidder is using the strategy $\beta_{i,0}$ for this first batch of auctions. The seller has then access to a dataset of $m$ bids $S_m^0 = \{b_1^0,...,b_m^0\}$ corresponding to these first $m$ auctions. She can then use a learning algorithm $\mathcal{M}_1$ to compute a new mechanism $\mathcal{M}_1(S_m^0)$   that would be used for the second stage. In full generality, at each stage, the bidder can change of bidding strategy and the seller can change of learning algorithm.  
} 
 \label{dynamic_game}
\end{figure*}
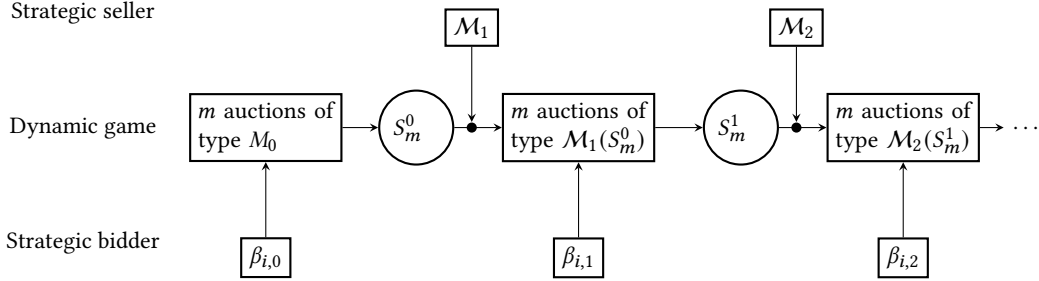
However,  bidders might have been  strategic in the first round, in order to  maximize their long-term utility.  Several approaches have been introduced for the seller to disincentive bidders from being strategic. A solution is to compute the  reserve price of a bidder using information stemming solely from the other ones \cite{ashlagi2016sequential,kanoria2017dynamic,epasto2018incentive}. This approach is theoretically sound, but practically limited as it requires that bidders have similar value distributions. For instance, it cannot handle heterogeneous settings with a dominant buyer \cite{epasto2018incentive} as the optimal reserve price of the latter can not be computed from bids of the others. Unfortunately, this is precisely the scenario where revenue optimizing mechanisms unveil their full potential. Another line of research assumes that bidders value the future a lot less than the seller \cite{amin2014repeated,golrezaei2018dynamic} by considering  discount factors of different magnitude orders. Although necessary to theory, this technique  introduces an artificial asymmetry between bidders and seller in order to force bidders to bid truthfully in the seller learning  phase (or at least during a significant fraction of it). With the same type of assumptions, \cite{braverman2018selling,deng2019prior} recently showed how the seller can extract the full welfare when bidders are using zero-regret type algorithms, leaving open the question of deriving good bidding strategies in such settings. 

If there does not exist assymetries between bidders and seller, bidders can actually adapt to automatic mechanisms in single-item auctions by being strategic in the first stage \cite{tang2016manipulate,nedelec2018thresholding,nedelec2019learning} . A new class of skewing or shading strategies was introduced for the lazy second price auction with monopoly reserve prices, and for the Myerson auction \cite{tang2016manipulate,nedelec2018thresholding}. The major and prohibitive drawback of these approaches is that they require that strategic bidders perfectly know the underlying mechanism design problem (i.e. the revenue maximization problem) solved by the seller. We aimed at exploring more dynamic gradient-based games where the seller and the buyers are using gradient-based algorithms.

Our starting point is a recent end-to-end learning approach  \cite{dutting2017optimal} computing revenue-maximizing auctions in various multi-item settings. This approach assumes that the seller can  generate samples from the value distributions of the bidders to update her mechanism. The mechanism is then parametrized by two neural networks corresponding respectively to the two rules, allocation and payment, that define a mechanism. These networks are trained to maximize the revenue of the seller under the incentive compatibility constraint. Inspired by the recent line of work that focuses on possible adversarial attacks on standard learning systems \cite{papernot2017practical}, we aim at  exploring manipulation opportunities for bidders in such learning approaches \cite{cai2015optimum,zhang2019samples}. 

Our contributions are the following. We introduce a new numerical framework to study economic interactions when several agents use learning algorithms based on data provided by other rational agents. We focus specifically on how bidders can find good bidding strategies when facing  mechanisms such as those introduced by \cite{dutting2017optimal}. From a game-theoretic standpoint, we cast the overall interaction between players as a Stackelberg game where bidders play first - hence are  leaders -, and the seller is the follower, playing second. We emphasize here that  we do not assume that bidders know the rules/algorithms/processes used by the seller to optimize her revenue; instead they discover them through a classic explore-exploit trade-off. We improve on recent single-item approach \cite{nedelec2019learning}, by removing the prior knowledge on  the exact algorithmic procedure used by the seller to optimize her mechanism. 

More precisely, we first solve the idealized setting where the seller is implimenting the exact Myerson auction corresponding to bidders' bid distributions. We then introduce some multi-agent gradient-based games between a seller, a strategic bidder and some non-strategic bidders. We consider the single-item and the multi-item setting. Inspired by reinforcement learning techniques, we introduce an exploration policy corresponding to a distribution over possible strategies and use classical policy optimization algorithms to tune the parameters of our policy. 
Furthermore, our approach elicits new shading strategies in classical and cutting-edge settings of the multi-item literature. 
For instance, we obtain a 54\% uplift in utility in the $2$ bidders and $2$ objects framework with one strategic bidder, where bidders have additive valuations and uniform value distributions between $0$ and $1$. 

This constitutes a first benchmark of the impact of strategic behavior on multi-item revenue maximizing auctions' performances. The implementation of our experimental setup in PyTorch is provided in the supplementary material.

\section{Auction design and Stackelberg games}
Classical mechanism design literature usually studies the Stackelberg game where the seller is the leader, and chooses a mechanism knowing the bidders' value distributions \cite{conitzer2006computing}. We assume that the seller does not have prior knowledge of bidders' value distributions and consider the reverse Stackelberg game where the bidders are leaders. Bidders are able to choose what value distribution to submit and thus impact the mechanism chosen by the designer.
\subsection{Notations}

\newcommand{\numItems}{\textsf{m}}
We consider the setting with a set of $n$ bidders and $\textsf{m}$ items with $M=\{1,\dots,\numItems\}$.  We denote by $v_i:\{0,1\}^M\to\mathbb{R}_{\geq 0}$ the valuation function of bidder $i$. For any bundle of items $S \subseteq M$, $v_i(S)$ represents how much bidder $i$ values the bundle $S$. As classically done in the auction literature, we assume that bidders' valuations are drawn independently from value distributions that we denote by $\{F_i\}_{i\in\{1,\dots,n\}}$. We denote by $F= F_1 \times \dots \times F_n$ their product distribution.  
\\
A bidding strategy $\beta_i:\mathbb{R}^{2^\numItems} \rightarrow \mathbb{R}^{2^\numItems}$ is a mapping from values to bid. We denote $\beta=(\beta_1,\dots,\beta_n)$ and write $\beta_{-i}$ the set of strategies without that of bidder $i$. The bid distribution $F_{B_i}$ is the distribution of bids induced by using $\beta_i$ on $F_i$. We denote by $\vec{b_i}$ the vector of bid submitted by bidder 
$i$ and $\mathcal{B}=\{(\vec{b_1},\dots,\vec{b_n}), \vec{b_i}\in\mathbb{R}^{2^\numItems}\}$ the set of all possible bid profiles.

A mechanism is a pair $m=(a,p)$ consisting of an allocation rule $a_i:\mathcal{B}\to \{0,1\}^{2^\numItems}$  and a payment rule $p_i:\mathcal{B}\to \mathbb{R}_{\geq 0}^{2^\numItems}$. For bids $\vec{b}=(\vec{b}_1,\dots,\vec{b}_n)$, $a_i(\vec{b})$ gives the allocation of the items, $p_i(\vec{b})$ the payment for each bidder and $u_i(\vec{b}) = a_i(\vec{b})(x_i-p_i(\vec{b}))$ the utility of bidder $i$.

The seller's revenue in an auction $(a,p)$ given bidding strategies $\{\beta_i\}_{i\in \{ 1,n\}}$ is defined as 
$$R(m,\beta) = \mathbb{E}_{F}\bigg(\sum_j a_j(\vec{B_1},\dots,\vec{B_n})p_j(\vec{B_1},\dots,\vec{B_n})\bigg)$$
where $\vec{B_i} = \beta_i(\vec{X}_i)$ and $\vec{X}_i$ is randomly drawn from $F_i$.
The utility of bidder $i$ is defined as : 
$$U_i(m,\beta) = \mathbb{E}_{F}\bigg(\left[X_i-p_i(\vec{B_1},\dots,\vec{B_n})\right]a_i(\vec{B_1},\dots,\vec{B_n})\bigg)$$

We will denote by $\beta_{Id}$ the truthful strategy corresponding to a player bidding his own valuation.
\begin{figure*}[t]
\centering

	\includegraphics[width=0.7\textwidth]{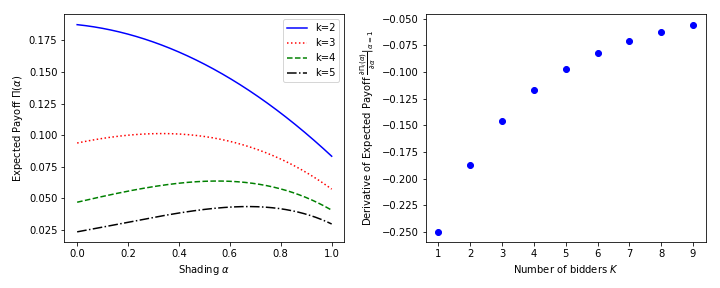}
	\caption{\textbf{Myerson auction : Expected payoff and its derivative for one bidder with linear shading} There are $K$ bidders with values $\Unif[0,1]$, only one of them is strategic. On the left hand side, we present a plot of the expected payoff $U_i(\alpha_i)$ of the strategic bidder for several values of $K$. On the right hand side, we present the derivative $\left.\frac{\partial U_i(\alpha)}{\partial \alpha}\right|_{\alpha=1}$ taken at the truthful bid ($\alpha=1$).}
    \label{fig:uniform_payoff_myerson}	
\end{figure*}

\subsection{Classical seller's learning problem and bidder's strategic answer}
We write the seller's mechanism optimization problem as a learning problem following methods introduced by the \emph{automated mechanism design} literature. Indeed, several works investigate methods for learning optimal mechanisms from data sampled from bidders' true value distributions, using numerical optimization and machine learning techniques.

In the classical framework, the seller seeks to solve the constrained optimization problem consisting of maximizing her revenue under the ex-post incentive compatibility constraint.
An automatic mechanism design algorithm $\mathcal{A}$ takes a class of mechanisms $\mathcal{M}$ and the bidders' value distributions as inputs, and outputs a mechanism solving a given constrained optimization problem.

 The problem of automated mechanism design as first introduced by \cite{conitzer2002complexity} and implemented in practice by \cite{OstSch11} essentially consists in a Stackelberg game where the seller takes bidders' value distributions as given, and enforces them to bid truthfully by choosing a DSIC mechanism. This first type of Stackelberg game takes the seller as leader.

 \begin{definition}[Seller/Bidder Stackelberg game] (Stackelberg) Game in which the seller chooses a mechanism among a class of DSIC mechanisms $\mathcal{M}$ which maximizes her revenue assuming she knows the bidders' value distributions.
 \end{definition}
From this game, we can define the seller's learning algorithm $\mathcal{A}$ solving this Seller/Bidder Stackelberg game. If we denote by $\mathcal{F}$ a class of value distributions and $\mathcal{M}$ a class of mechanisms, a seller's learning algorithm $\mathcal{A}$ is defined as
\begin{align*}
\mathcal{A}:\hspace{0.5cm}  F \hspace{0.5cm}  &\mapsto m(F) = \argmax{m \in \mathcal{M}} {R(m,\beta_{Id})} 
\\&\hspace{0.5cm} s.t.\;\;U_i(m,\beta_i,\beta_{-i}) \leq  U_i(m,\beta_{Id},\beta_{-i}),
\\&\hspace{1.3cm}\forall i \in \{ 1,n\}, \forall \beta_i, \forall \beta_{-i}
\end{align*}
In practice, the fact that the seller uses past bids to estimate bidders' value distributions before optimizing her mechanism in repeated auctions provides bidders with the opportunity to design ``attacks" in order to find bidding strategies that increase their long-term utility. By strategically adjusting their bids, they are able to control the bidding distributions perceived by the seller and used to optimize her mechanism. This corresponds to a new Stackelberg game in which the bidders are the leaders of the game, which is the focus of our work here.

\begin{definition}[Bidder/Seller Stackelberg game]\cite{tang2016manipulate,nedelec2019learning} Stackelberg game in which strategic bidders assume the existence of a seller's learning algorithm $\mathcal{A}$. Each strategic bidder $i$ chooses a strategy $\beta_i$ that induces a (pushforward) bid distribution $F_{B_i} = \beta_i\#F_i$ used as input by the seller's algorithm. The goal of the strategic bidder is to optimize
$$\argmax{\beta_i \in \mathcal{B}_i} {U_i(\mathcal{A}(F_{B_i},F_{-i}),\beta_{i})}$$
\end{definition}

Several approaches have already tackled this problem \cite{tang2016manipulate,nedelec2018thresholding,nedelec2019learning}. In all these papers, the authors assume perfect knowledge of the optimization algorithm used by the seller. Our goal is to extend these approaches by getting rid of the assumption that bidders know the seller's algorithm and by proposing a method that automatically adapts to this new framework.
We provide a general method that applies in particular to general multi-item auctions, and hence to cutting edge auction theoretic results. For these auctions, allocation and payment rules are currently available only through numerical methods such as the one developed by \cite{dutting2017optimal}, which preclude the design of attacks based on analytic understanding of auction rules. We provide a general approach to designing such attacks, proving that the networks introduced by \cite{dutting2017optimal} are not robust to adversarial attacks.

These adversarial ``attacks" could be called Stackelberg responses to black-box automatic mechanism design. They exploit a conceptual opening in most automatic mechanism design works, i.e. the breakdown  of incentive compatibility for the buyer when the seller optimizes over incentive compatible auctions. As such, they differ from standard adversarial attacks in e.g. computer vision, which generally rely on the lack of local robustness of a classifier. Two other features are notable: these ``attacks'' do not necessarily yield lower revenues for the seller \cite{nedelec2018thresholding}; and they are also part of a dynamic game between buyers and seller and as such have a dynamic component that is absent from classical and static machine learning frameworks, such as image classification. 

 \section{An analytical solution to the Myerson Stackelberg game}

To get a sense of what would be the optimum in the perfect information setting, we first focus on the idealized Stackelberg game where the bidder assumes that the seller is using the Myerson auction. This Myerson auction corresponds to the bid distribution induced by the strategic bidder during the seller's learning stage. 

We can use \textit{the Myerson lemma} and show that the expected utility of the strategic bidder using the bidding strategy $\beta$ in the Myerson auction is
$$
U_i(\shadingFunc_i)=\Expb{[X_i-h_{\beta_i}(X_i)] F_Z(h_{\beta_i}(X_i))}\;.
$$
with $F_Z$ the cumulative distribution function of $$Z=\max_{2\leq j \leq K}(0,\vValue_j(X_j)),$$ $X_i$ is the value of bidder $i$, and $h_{\beta_i} = \psi_{B_i}(\beta_i(X_i))$ is the virtual value function associated with the bid distribution.
Suppose that $\shadingFunc \mapsto \shadingFunc_t=\shadingFunc +  t \rho$, where $t>0$ is small and $\rho$ is a function. We note that $h_{\beta+t\rho}(x)=h_\beta+t h_\rho$. We have the following result.

\begin{lemma}\label{lemma:directDerivativeMyersonAuction}
Suppose we change $\beta$ into $\beta_t=\beta+t \rho$. Both $\beta$ and $\beta_t$  are assumed to be non-decreasing. Call $x_{\beta}$ the reserve value corresponding to $\beta$, assume it has the property that $h_\beta(x_\beta) = 0$ and $h_\beta'(x_\beta)\neq 0$ ($h_\beta'$ is assumed to exist locally). Assume $x_\beta$ is the unique global maximizer of the revenue of the seller. Then, 
\begin{align*}
\hspace{-7pt}\left.\frac{\partial}{\partial t} U(\beta_t) \right|_{t=0} 
&=\Expb{h_\rho(X) [(X-h_{\beta}(X))f_Z(h_{\beta}(X))-F_Z(h_{\beta}(X))]\indicator{X>x_{\beta}}}
\\
&+\frac{h_\rho(x_{\shadingFunc})}{h'_{\shadingFunc}(x_{\shadingFunc})}\prod_{i=2}^K F_{V_i}(0) f_{1}(x_{\shadingFunc}) x_{\shadingFunc}\;,\notag
\end{align*}
\end{lemma}
\begin{proof}
Taking directional derivative of the utility of the bidder gives the equation.
\end{proof}
As introduced in \cite{nedelec2019learning}, there exists a simple relationship between the virtual value and the bidder's strategy. 
\begin{lemma}\label{definition_psi}
Suppose $B_i=\beta_i(X_i)$, where $\beta_i$ is increasing and differentiable and $X_i$ is a random variable with cdf $F_i$ and pdf $f_i$. Then 
\begin{equation*}\label{eq:ODEPhiG}
h_{\beta_i}(x_i)\triangleq \beta_i(x)-\beta_i'(x)\frac{1-F_{i}(x)}{f_{i}(x)} = \psi_{F_{B_i}}(\beta_i(x)) \;.
\end{equation*}
\end{lemma}
Using these directional derivatives and the relationship between the virtual value of the induced bid distribution and the bidder's strategy, we can derive what are the optimal linear strategies in the Myerson auction. 

Though we do not need symmetry of the bidders' value distribution, we start by a few examples assuming it for concreteness.  We recall that if $F$ is the cdf of $X_i$, $G(x)=F^{n-1}(x)$ in the case where we have $n$ symmetric bidders. 
\\
\textbf{Example of uniform [0,1] distributions:} 
In this case, $\vValue_i(x)=2x-1$ on [0,1] and $\vValue_i^{-1}(0)=1/2$. Also, $G(x)=x^{n-1}$. Then, using for the instance the representation of the derivative of $U_i(\alpha)$ appearing in the proof of Lemma \ref{lemma:directDerivativeMyersonAuction}, we have  
\begin{align*}
\left.\frac{\partial U_i(\alpha)}{\partial \alpha}\right|_{\alpha=1}&=
\int_{1/2}^1 \left(x-\frac{1}{2}\right)[(n-1)-x(n+1))]x^{n-2} dx  
    \\&=-\frac{1}{n2^{n+1}}(
    2^n-1)<0\;.
\end{align*}
Hence, each user has an incentive to shade their bid. We note that the derivative goes to 0 as $n\tendsto \infty$ (see also Fig.\ref{fig:uniform_payoff_myerson} right side), which can be interpreted as saying that as the number of users grows, each user has less and less incentive to shade. We can also observe on Fig.\ref{fig:uniform_payoff_myerson} (left side) that the difference between the payoff at optimal shading $\alpha^*$ and the payoff without shading -- $(U(\alpha^*) - U(0))$ -- decreases with $K$.
Indeed, when $K$ grows, the natural level of competition between the bidders makes the revenue optimization mechanisms (e.g. dynamic reserve price) less useful. Logically, being strategic against it in such case does not help much.
For very few bidders, the contrary happens. 
For $K=2$, we even observe that the optimal strategy is to bid with a shading of $\alpha = 0^+$ to force a price close to $0$ while still winning with probability $1/4$ -- when one is beating his reserve and the opponent is not beating his, with the result of almost doubling the payoff. 

With more advanced arguments and stronger assumptions on bidders' value distributions, we can derive what is the optimal best response for a strategic bidder in a large class of value distributions called the Generalized pareto distribution.
\begin{definition}
\item The family of Generalized Pareto distributions, parametrized by $(\mu,\sigma,\xi)$ where $\sigma >0$ and $\xi\leq 0$, has distribution 
\begin{equation*}
F_{\mu,\xi,\sigma}(x)   = \begin{cases} 1-(1 +\frac{\xi (x-\mu)}{\sigma})^{-1/\xi} & \text{for } \xi < 0 \\
1-e^{-(x-\mu)/\sigma} & \text{for } \xi= 0\end{cases}
\end{equation*} and its virtual value is affine \citep{balseiro2020multistage}
$$\psi_{\mu,\xi,\sigma}(x) = (1-\xi)x+\xi\mu-\sigma$$
\end{definition}
 \begin{figure*}[t]
\begin{center}
	\includegraphics[width=0.7\textwidth]{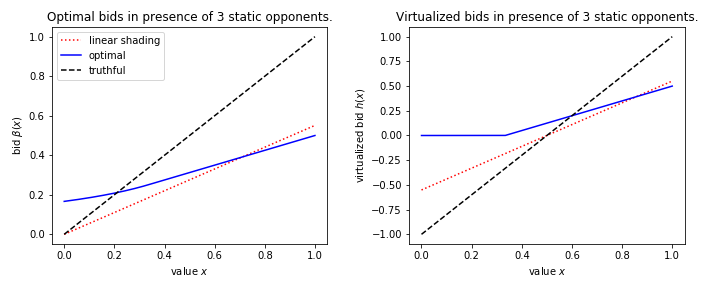}
\end{center}
	\caption{\textbf{Myerson auction: Bids and virtualized bids with one strategic bidder} There are K=4 bidders, only one of them is strategic. On the left hand side, we present a plot of the bids sent to the seller. ``Linear shading" corresponding to a bid $\shadingFunc_\alpha(x)=\alpha x$, where $x$ is the value of bidder 1; here $\alpha$ is chosen numerically to maximize that buyer's payoff. ``Optimal" corresponds to the strategy described in Theorem 1, with $\eps=0^+$. On the right hand side (RHS), we present the virtualized bids, i.e. the value taken by the associated virtual value functions evaluated at the bids sent to the seller.}\label{fig:uniform_bid_profiles}	
\end{figure*}
\begin{theorem}\label{lemma:OneStrategicMyersonInGP}
Suppose a strategic bidder faces $K-1$ opponents sharing the same distribution $F_Y$ in the Generalized Pareto family. Then, assuming that the seller is welfare benevolent, her optimal shading function is such that her virtualized bid $h_{\text{optimal}}(x)$ satisfies 
$$
h_{\text{optimal}}(x)=\max(0,\psi_Y(\beta^I_{Y,K}(\psi_Y^{-1}(x))))\;.
$$
where $\beta^I_{Y,K}$ is the first price bid of a bidder facing competition with cdf $G=F_Y^{K-1}$.
The corresponding shading function $\beta$ can be easily obtained by an application of Lemmas \ref{lemma:directDerivativeMyersonAuction} and \ref{definition_psi}.
\end{theorem}
\begin{proof}
We need to find $h$ that maximizes
$$
\int_{x:h(x)\geq 0}(x-h(x))\bigg[F_Y\left(\frac{h(x)}{c_{\psi_Y}}+r_Y^*\right)\bigg]^{K-1} f_1(x) dx\;.
$$
We can maximize point by point and hence we are looking for $t^*(x)$ such that
$$
t^*(x)=argmax_t(x-t)F_Y\left(\frac{t}{c_{\psi_Y}}+r_Y^*\right)^{K-1}\;, t>0\;.
$$ 
Differentiating the above expression gives 
$$
\delta(t)=f_Y(t/c_{\psi_Y}+r)F_Y^{K-2}(t/c_{\psi_Y}+r^*_Y)\left[\frac{x-t}{c_{\psi_Y}}-\frac{1}{K-1}
\frac{F_Y}{f_Y}(t/c_{\psi_Y}+r^*_Y)\right]\;.
$$
The expression in the bracket can be written $\psi_Y^{-1}(x)-H(\psi_Y^{-1}(t))$ where $H=\textrm{id}+\frac{G_{Y,K-1}}{g_{Y,K-1}}$, where $G_{Y,K-1}$ is the cdf of the max of $K-1$ i.i.d random variables and $g_{Y,K-1}$ its derivative. Elementary computations show that this function is increasing in GP families. In fact its derivative can be shown to be $1+(1-\frac{\xi}{\sigma}G_{}(x)/(1-G(x)))/(K-1)$ and $\xi<0$. Hence $H(\psi_Y^{-1}(t))$ is also increasing. Hence $\delta(t)$ is a decreasing function of $t$. It is also trivially continuous in GP families. We conclude that the equation $\delta(t)=0$ has at most 1 positive root. 

If $\psi_Y^{-1}(x)<H(\psi^{-1}_Y(0))$, we see that $\delta(t)<0$ for $t\geq 0$, in which case $t^*=0$. 
If that is not the case, then $\psi^{-1}_Y(t^*)=H^{-1}(\psi_Y^{-1}(x))$. Hence we have shown that 
$t^*=\max(0,\psi_Y(H^{-1}\psi_Y^{-1}(x)))$. Now we notice that the $H^{-1}(x)$ is nothing but the first price bid of a bidder facing competition with cdf $G=F_Y^{K-1}$, a bid function we denote by $\beta^{I}_{Y,K}$. So we conclude that 
$$
h_{\text{optimal}}(x)=\max(0,\psi_Y(\beta^I_{Y,K}(\psi_Y^{-1}(x))))\;.
$$
Once again the fact that $h_{\text{optimal}}$ is non-decreasing (as a composition of non-decreasing functions) avoids issues related to ironing.
\end{proof} 
The strategies are shown in Figure \ref{fig:uniform_bid_profiles}. This solves the bidder/seller Stackelberg game when the seller is using the Myerson auction on the bid distrbution observed during her learning stage.

 \section{Gradient-based Stackelberg games between bidders and seller}
To extend the idealized setting to more realistic assumptions, we now assume that instead of computing directly the Myerson auction, the seller is a using a gradient-based learning mechanism. We consider the approach taken by \cite{dutting2017optimal} for the implementation of the seller's optimization process. Their work provides a general algorithmic approach to approximately solve this problem in multi-item, multi-bidder settings. The seller's auction is parametrized by a weight vector $\omega$ corresponding to two neural networks which take bids for each item  and each player ($n\times m$ entries) as inputs and return respectively the allocation probability $a_{\omega}$ of each item, for each player ($n\times m$ outputs) and the payment for each player $p_{\omega}$ ($n$ outputs). In the case of combinatorial auctions, bidders would submit a bid for each possible bundle ($n\times 2^m$ entries).

For the single-item setting, we consider the MyersonNet architecture \cite{dutting2017optimal}. The allocation rule is defined as an invertible neural network parametrizing a transformation of the bid. The payment rule is obtained in such a way that the auction is DSIC following the Myerson lemma. This provides a first benchmark on how seller learning algorithms are sensitive to adversarial attacks. We focus on one specific bidder and assume that the strategies of other bidders are fixed.  We show how the strategic bidder can optimize an exploration bidding policy to increase his utility when the seller is using a MyersonNet-type architecture to optimize her selling mechanism. 

\begin{definition}[Exploration bidding policy]
We consider a set of possible bidding strategies $\mathcal{B}$. An exploration bidding policy $U$ is a distribution over this set of strategies.
\end{definition}

We first consider the case where $\mathcal{B}$ is the set of linear bidding strategies because of their simplicity and wide use in modern industrial bidding engines. To parametrize our exploration bidding policy, we use a normal distribution such that 
$$\lambda \sim \mathcal{N}(\mu,\sigma^2) = U(\mu,\sigma^{2})$$
with corresponding bidding strategy $\beta_\lambda(x) = \lambda x$ for the strategic bidder. We do not require any assumption on the other bidders' behavior. 

According to the exploration policy, we sample several shading parameters $\lambda$ which are used as bid multipliers by the strategic bidder. The goal of the strategic bidder is to optimize the parameters $\mu$ and $\sigma^2$ to maximize his utility when the seller is using the MyersonNet architecture. A representation of the global architecture is provided in Figure \ref{global_architecture}.
\\
In \cite{nedelec2018thresholding}, they introduced a class of functions which are optimal in several types of revenue-maximizing auctions. We also consider this class of strategies in our experiments in a second time. The thresholded strategies they introduced can be parametrized by three parameters~: the threshold $r$ corresponding to the value below which the virtual value is thresholded; the slope $a$ of the bidding distributions' virtual value after the threshold $r$; and the value $\epsilon$ of the virtual value before $r$: in the case of a uniform value distribution, this gives a bidding strategy parametrized such that the virtual value of the bid distribution satisfies $\psi_{B}(x) = \epsilon$ for $x<r$ and $\psi_{B}(x) = ax - r$ for $x\geq r$:
 We maintain a normal distribution with diagonal covariance $\Sigma$  over $\Lambda=(r,a,\epsilon)$ and optimize the exploration bidding policy corresponding to this class of strategies. We show that this results in a large increase in terms of utility for the strategic bidder, without needing to know the exact optimization procedure of the seller. In practice, we would only need to detect when the learning stage of the seller is finished, which could be adressed in another following work.

\begin{figure}[t!]
\begin{tikzpicture}[
    block/.style={rectangle, draw, text centered, rounded corners,     minimum
    height=4em,text width=5.5em}
]
	\filldraw[fill=red!20]  (3.2,-0.3) rectangle (5.6,2.8); 
	\filldraw[fill=green!20] (0.4,-0.8) rectangle (2.3,0.8); 

	\node[circle, draw, thick] (z) {$\vec{x}_i$};
	\node[above=of z, circle, draw, thick] (xinit) {$\vec{x}_j$};
	
	\node[above=of xinit, circle, draw, thick] (xinit2) {$\vec{x}_k$};

	\node[rectangle, draw, thick, right=0.5em of z, text width=1.5cm] (bidderbox) {strategic bidder network} ;
	\node[rectangle, draw, thick, right=0.5em of bidderbox] (x) {$\vec{b}_i$};

	\node[above=of x, circle, draw, thick] (xt) {$\vec{b}_j$};
	\draw[-stealth, thick] (xinit) -- node[above] {fixed strategy} (xt);
	\node[above=of xt, circle, draw, thick] (xt2) {$\vec{b}_k$};
	\draw[-stealth, thick] (xinit2) -- node[above] {fixed strategy} (xt2);

	\draw[-stealth, thick] (z) -- (bidderbox);

	\draw[-stealth, thick] (bidderbox) -- (x);

	\node[rectangle, draw, thick, right=2em of x, yshift=6.5em, text width=1.5cm] (D) {allocation network} ;
	\node[rectangle, draw, thick, below=2.5em of D,text width=1.5cm] (paymentnet) {payment network};

	\node[right=3em of D] (out) {revenue};
	\node[right=3em of D, yshift=-2.5em] (out2) {utility};

	\draw[-stealth, dashed] (D) --  (out);
	\draw[-stealth, dashed] (D) --  (out2);
	\draw[-stealth, dashed] (paymentnet) --  (out);
	\draw[-stealth, dashed] (paymentnet) --  (out2);

	\draw[-stealth,dashed]    (out2) to[out=-3,in=-30] (1.0,-0.7);
	\draw[-stealth,dashed]    (out) to[out=3,in=20] (5.0,3.0);

	\draw[-stealth, thick] (xt) -- node[above] {} (D);
	\draw[-stealth, thick] (xt) -- node[above] {} (paymentnet);
	\draw[-stealth, thick] (xt2) -- node[above] {} (D);
	\draw[-stealth, thick] (xt2) -- node[above] {} (paymentnet);
	\draw[-stealth, thick] (x) -- node[above] {} (D);
	\draw[-stealth, thick] (x) -- node[above] {} (paymentnet);

			
			
			

\end{tikzpicture}
\caption{\textbf{General architecture for adversarial learning in revenue-maximizing auctions.}. The green box corresponds to the strategic bidder's parameters, the red box corresponds to the parameters that the seller can optimize. The figure represents three bidders with one strategic bidder optimizing his bidding strategy with assuming the other bidders are using some fixed strategies.} 
\label{global_architecture}
\end{figure}
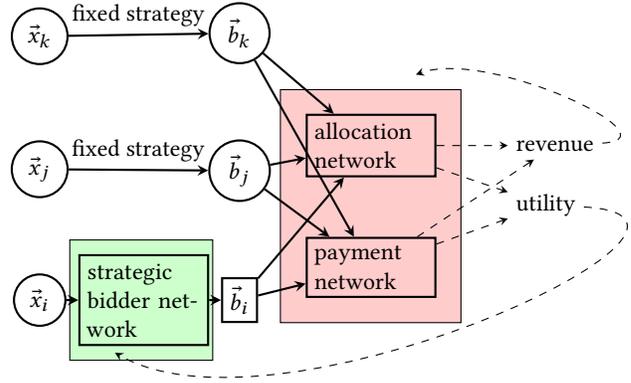

\begin{table*}[t]
\centering
\small

\begin{tabular}{|c|c|c|c|c|c|c|c|c|c|c|}
\hline
\multicolumn{1}{|c|}{\multirow{3}{*}{Setup}} & \multicolumn{2}{c|}{\multirow{2}{*}{Truthful VCG}} & \multicolumn{2}{c|}{\multirow{2}{0.075\textwidth}{Truthful Myerson}} & \multicolumn{6}{c|}{\textit{MyersonNet}}                                                                            \\ \cline{6-11} 
\multicolumn{1}{|c|}{}                       & \multicolumn{2}{c|}{}           &\multicolumn{2}{c|}{}                                                                   & \multicolumn{2}{c|}{Truthful}       & \multicolumn{2}{c|}{Linear} & \multicolumn{2}{c|}{Thresholded} \\ \cline{2-11} 
                    & \textit{utility}       & \textit{revenue}        & \textit{utility}       & \textit{revenue}      & \textit{utility}       & \textit{revenue} & \textit{utility} & \textit{revenue} & \textit{utility}   & \textit{revenue}   \\ \hline
K=2                &0.168&          0.33       & 0.083                   &     0.416                                                                                     & 0.083           & 0.417                 & 0.169            &      0.304            & 0.181             &   0.368                 \\ \hline
K=3               & 0.083&          0.500       & 0.057              &     0.531                                                                                    & 0.057             &   0.530               & 0.100             &      0.46            & 0.115              &     0.495               \\ \hline
K=4              & 0.05&       0.60            &          0.040              &         0.612                                                                            &     0.040             & 0.612                 &     0.064             &        0.570          &   0.069                 &    0.587                \\ \hline
\end{tabular}
\caption{\textbf{Experiments for the single-item setting. All bidders have a uniform value distribution on $[0,1]$. The strategic bidder is playing against $K-1$ bidders who are bidding truthfully. The seller is using the MyersonNet architecture as learning algorithm.} To compute the performance of the exploration policy, we average over $q=50$ strategies sampled from the exploration policy.}
\label{table_single_item}
\end{table*}

We use the classical Reinforce algorithm \cite{williams1992simple} to optimize the parameters $\Lambda$ of the distribution. In the experiments, we do not optimize the variance of the distribution (hence it never tends to 0) to continue the exploration. In practice, it makes the approach robust to any change in the seller's optimization procedure as the bidder never stops exploring. If  the goal were solely to find the optimal bidding strategies by enabling the variance to converge to zero, we could use classical evolutionary search algorithms such as NES \cite{wierstra2008natural} to also optimize the variance of the distributions. The full procedure is presented in Algorithm \ref{Algo:adv_train}. All our implementations are provided in Pytorch.

\begin{algorithm}
\caption{Adversarial training for sellers' learning mechanisms}
\label{Algo:adv_train}
  \SetAlgoLined
  \SetKwInOut{Input}{Input}\SetKwInOut{Init}{Initialization}
 \Input{Distributions $F_1,\dots,F_n$, seller's learning mechanism $\mathcal{A}$}
 \Init{Initialize $\mu_1,\;\Sigma$\;}
  \For{$t = 1$ to $T$} {
  \For{$l = 1$ to $L$}
   {  Sample $\Lambda_{i,l}\sim\mathcal{N}(\mu_t,\Sigma)$\\

   Run subroutine $\mathcal{A}$ to optimize seller's mechanism on bids induced by $\beta_{\Lambda_{i,l}}$ for the strategic bidder $i$ and $\beta_j=\beta_{Id}$ for other bidders \\
   Compute strategic bidder's utility: $U_{\Lambda_{i,l}} = U(\mathcal{A}(F_{B_{\Lambda_{i,l}}},F_{-i}),\beta_{\Lambda_{i,l}})$\\}
   Compute gradient: $\nabla U(\mu_t) = \frac{1}{L} \sum_{l=1}^L  U_{\Lambda_{i,l}}\log\left(f_{\mathcal{N}(\mu_t,\Sigma)}(\Lambda_{i,l})\right)$
   Update: $\mu_{t+1}\leftarrow \mu_t-\rho_t \nabla U(\mu_t)$}
\end{algorithm}
 
The optimization procedure is the following. The goal of the algorithm is to maximize the expected utility of the strategic bidder~: 
$$argmax_{\mu \in \mathbb{R}} U(\mu) = \mathbb{E}_{\Lambda\sim U(\mu,\Sigma)}\bigg(U(m(\beta_\Lambda),\beta_{\Lambda})\bigg)$$
where $U$ is the strategic bidder's utility and $m(\beta_\Lambda)$ is the mechanism resulting from a training where the neural networks take bids  $b_i = \beta_\Lambda(v_i)$ as inputs, and where $v_i$ is sampled from  $F_i$. We use the MyersonNet architecture for the single-item case and the RegretNet architecture for the multi-item setting.  In both cases, they takes bids induced by the shading strategy as inputs. We sample several shading parameters according to the exploration bidding policy and take a gradient step according to: 
$$
\nabla_\mu U(\mu)= \textbf{E}_{\Lambda \sim U(\mu,\Sigma)}\left(U(\beta_{\Lambda})\frac{\nabla_\mu p_{\mu,\Sigma}(\Lambda)}{p_{\mu,\Sigma}(\Lambda)}\right)\;.
$$
with $p_{\mu,\Sigma}$ the probability density function (henceforth pdf) corresponding to $U(\mu,\Sigma)$. To compute $U(\beta_\Lambda)$ we run a full training of the MyersonNet architecture.

\subsection{Extension to various multi-item settings}
Our approach can easily be extended to the multi-item setting where no theoretical solution to the bidder/seller Stackelberg game is known. It offers a first benchmark on how bidders can manipulate such auctions. To study multi-item settings, we use the RegretNet architecture \cite{dutting2017optimal} to parametrize the sellers' mechanism instead of the MyersonNet, used in the single item case. RegretNet uses of two feed-forward deep neural networks to parametrize respectively the allocation and payment rules. Both networks contain $H$ hidden layers of size $h$ activated with $\tanh$. We limit our study to additive bidders, i.e. bidders whose valuation for a bundle $S$ is $v(S)=\underset{k\in S}{\sum}v(k)$. To enforce that each item is allocated at most once, we use a softmax function on the output layer (size $n\times m$) so that $\forall i \in M,\;\underset{k\in \{ 1, n \}}{\sum} a_{k,i} \leq 1$, where $a_{k,i}$ is the probability for bidder $k$ to get object $i$ outputed by the allocation network. The output layer of the payment network is activated using the sigmoid function, and outputs $n$ coefficients $\bar{p_k}\leq 1$ which combined with the allocation network output return the actual payment for each bidder: $\forall k\in\{ 1,n\},\;p_k = \bar{p_k}\underset{l\in M}{\sum} a_{k,l}\vec{b_{k,l}}$. The condition $\bar{p_k}\leq 1$ ensures that expected payment never exceeds expected gains for the seller (\emph{indidual rationality} condition).

We benchmark the impact of a linear exploration policy on the RegretNet architecture and see how the seller's revenue is impacted by a strategic bidder in the multi-item setting. We consider two classical settings of the multi-item literature. We denote by Setting I the setting with two items and two bidders with additive valuations and uniform value distribution $F_1=F_2=U([0,1]^2)$; and by Setting II the setup with two objects and three bidders with additive valuations and value distribution $F_1=F_2=F_3=U([0,1]^2)$. 

\subsection{Handling the exploration stage of the seller}
The tradeoff between exploration and exploitation from the seller standpoint was introduced in \cite{amin2013learning} and refined in \cite{mohri2015revenue, golrezaei2018dynamic}. They introduce a parameter $\alpha$, $0\leq \alpha<1$, to define this trade-off, assuming the ratio of length between the first and the second stage is equal to $\alpha/(1-\alpha)$. In \cite{amin2013learning}, they show that if bidders are non-discounted buyers, there must exist a good strategy for them in this mechanism, forcing the seller to suffer a regret linear in the number of auctions. We can derive such strategies by adding bidder's utility in the first stage where the bidder is using his strategy in a non-optimized auction such as a second price auction without reserve price. To consider the cost of using a certain strategy in the first stage of the game, we can add it to the objective function and optimize: 
 $$U_\alpha =  \alpha U_{\text{second price without reserve}}(\beta) + (1-\alpha) U_{\text{MyersonNet}}(\beta)$$
Again, we assume that bidders are the leaders in this framework since they know the mechanism used by the seller, the length of the exploration stage and can choose their strategy accordingly.

\section{Experimental results}

\begin{table*}[t]
\small
\centering
\begin{tabular}{|l|l|l|l|l|l|l|}
\hline
\multicolumn{1}{|c|}{\multirow{2}{0.2\textwidth}{Setting}} & \multicolumn{2}{c|}{VCG} & \multicolumn{2}{c|}{RegretNet (truthful)} & \multicolumn{2}{c|}{RegretNet (adversarial)} \\ \cline{2-7} 
\multicolumn{1}{|c|}{} & \multicolumn{1}{c|}{\textit{utility}} & \multicolumn{1}{c|}{\textit{revenue}} & \multicolumn{1}{c|}{\textit{utility}} & \multicolumn{1}{c|}{\textit{revenue}} & \multicolumn{1}{c|}{\textit{utility}} & \multicolumn{1}{c|}{\textit{revenue}} \\ \hline
Setting Ia : &  &  &  &  & &  \\
two bidders, two objects & 0.336 & 0.666 & 0.149 & 0.882 & 0.306 (+108\%) & 0.696 (-21\%) \\ 
uniform value distribution &  &  &  &  &  &  \\
\hline
Setting Ib : &  &  &  &  & &  \\
two bidders, two objects & 1.000 & 1.000 & 0.504 & 1.481 & 0.574 (+13\%) & 1.443 (-2.5\%) \\ 
exponential
 value distribution &  &  &  &  &  &  \\
\hline
Setting IIa : &  &  &  &  &  &  \\ 
three bidders, two objects & 0.166 & 1.000  & 0.096 & 1.034 & 0.148 (+54\%) & 0.985  (-4.7\%) \\ 
uniform value distribution &  &  &  &  &  &  \\ \hline
Setting IIb : &  &  &  &  &  &  \\ 
three bidders, two objects & 0.666 & 1.666  & 0.249 & 1.804 & 0.294 (+18\%) & 1.801 (-0.1\%) \\ 
exponential value distribution &  &  &  &  &  &  \\ \hline
\end{tabular}
\caption{\textbf{Experiments for the multi-item setting. The strategic bidder is using a linear bidding exploration policy with parameter $\sigma^2_k = 0.05$. The seller is using the RegretNet architecture as selling mechanism.} We run $T=150$ adversarial training epochs,and base our evaluation on averaging over $q=12$ strategies from the exploration policy.}
\label{table_regretnet}
\end{table*}
Our pipeline of experiments provides a first benchmark on the impact of adversarial attacks of well-known seller's learning mechanisms. We consider the uniform distribution on $[0,1]$ since this is the standard textbook example in auction design and the exponential distribution.  We use $\sigma^2=0.005$ in our experiment to learn the linear shading and $\Sigma=\text{diag}(0.005,0.005,0.005)$ to learn the parameters of the thresholded exploration policy. To compute strategic bidder's utility and seller's revenue, we sample bidding strategy parameters according to the exploration bidding policy.  Our result are reported on Table \ref{table_single_item}. For the setting with three bidders, we get an uplift of 20\% in terms of utility for the strategic bidder and a decrease of 9\% in seller's revenue with a simple linear shading policy. It shows as expected that the MyersonNet architecture is not robust to adversarial attacks from a strategic bidder. Interestingly, with the thresholded strategies in the case of two bidders, the exploration bidding policy leads to both a higher utility for the strategic bidder \emph{and} a higher revenue for the seller than when using linear shading. This is the illustration that the auction game is not a zero-sum game between the seller and the buyers.

We compare the performance of our approach with several natural baselines. The Vickrey-Clark-Gloves (VCG) auction corresponds to the second-price auction without reserve price. This is a welfare-maximizing auction. Possibly surprisingly, it is possible to get a higher utility for a strategic bidder when seller is using a revenue-maximizing auction rather than a welfare-maximizing auction. Indeed, Myerson reduces the competition when all the other bidders are bidding below their reserve price. The strategic bidder takes advantage of this reduction of competition to increase his utility. We only provide experiments for less than 4 bidders since the interest of revenue-maximizing auctions both in terms of utility and revenue decreases dramatically with the number of bidders when they all have symmetric value distributions. Our architecture could also enable to study the impact of other strategic buyers by running at the same time several buyers' learning algorithms.
\subsection{Experiments with multi-item auctions}
Using simple linear shadings in multi-item settings yielded considerable improvements in bidders' utility. We implemented Algorithm \ref{Algo:adv_train} initializing $\mu_1$ to be an array of $m$ ones (corresponding to the thuthful strategy), and $\sigma^2_k = 0.05$ for all $k\in M$. We run $T=150$ adversarial training epochs, and sample $q=12$ lambdas per epoch. We optimize the seller mechanism every $3$ adversarial epoch by training the RegretNet architecture. We implement the RegretNet architecture in PyTorch by using two neural networks with $H=2$ hidden layers of size $h=30$. Our experimental results are reported in Table \ref{table_regretnet}. We observe substantial improvements in bidders' utility, with a 108\% uplift for Setting Ia and a 54\% uplift for Setting IIa. This is the performance of the exploration and it would be possible to improve the strategic bidder's utility by decreasing the variance of the exploration policy at the cost of not being robust to changes of the learning mechanism. This suggests that even better improvements in utility could be found using more complex bidding strategies in the spirit of the thresholded-virtual-value strategy introduced by \cite{nedelec2018thresholding} for the single-item framework. 

Our work thus opens the door to several natural extensions such as using neural networks to parametrize more complex bidding strategies, or studying other bidder types, valuation distributions and auctions such as the combinatorial auction. However, training neural networks to learn the exploration policy would increase the running time of the procedure, which is already substantial for linear shading strategies. This provides a first benchmark to design adversarial attacks against sellers' learning algorithms. This benchmark could be extended in the near future by testing new seller algorithms and new architecture to learn strategic behaviors. This reinforces the idea that the conceptual mistake of not treating the game where the seller uses past bids to optimize the auction as a Stackelberg game can be very costly for bidders. Moreover, they show that data-driven automatic mechanisms are vulnerable to adversarial attacks, hence providing motivation for practical implementation of adversarial attacks on modern marketplaces, or implementation of automatized mechanisms robust to adversarial attacks on these same platforms.

\section{A need for adversarially-robust seller learning mechanisms}
A natural extension to the design of adversarial attacks against data-driven automated selling mechanisms is the design of learning algorithms which are robust to adversarial attacks. This line of work has been initiated by \cite{allouah2018prior}, who find mechanisms which maximize the seller's revenue against the worst bid distribution in a certain class. To avoid dealing with worst-case scenarii, an intermediate approach would be to consider mechanisms robust to a class of bidding strategies and a class of initial value distributions. 
\begin{definition}[$\epsilon$ adversarially-robust learning algorithm]
A selling learning algorithm $\mathcal{M}$ is said to be $\epsilon$ adversarially-robust for this class of value distributions, if for any value distributions $F_i$ in this class, for any adversarial attack $\beta^*$, with ,$\beta^{Tr}$ the truthful strategy, the seller's revenue $R$ when the strategic bidder is using U verifies
$
R(\mathcal{M}(F_i,\beta^*),\beta^*) \geq R(\mathcal{M}(F_i,\beta^{Tr}),,\beta^{Tr}) - \epsilon \;.
$
\end{definition}
This leads to a new definition of incentive compatible learning algorithms where bidders have an incentive to bid truthfully even if the seller is using past bids to optimize her mechanism. A follow up on our work could be to investigate feasibility of such robust mechanisms by adding a constraint to an augmented Lagrangian method similar to that used by \cite{dutting2017optimal}. Our approach is the first necessary step in  the design of such robust mechanisms since it computes how the revenue is impacted when using a given learning mechanism. 
\section{Conclusion} 
We present a new way to design adversarial attacks against cutting-edge automatic mechanism design algorithms. Our approach yields very substantial utility gains for the strategic bidder in our numerical experiments. This allows buyers to quantify the price of revealing information about their values in repeated auctions. From a theoretical standpoint, this offers a new tool to study economics interactions through an algorithmic lens and represents a new step to reinterpret economics problems as algorithmic learning problems between strategic agents.


\clearpage
\bibliographystyle{ACM-Reference-Format} 
\balance
\bibliography{bidderDependentAuctionsAll}


\end{document}